\documentclass[12pt]{article}
\usepackage{amsfonts,amssymb,amsmath}

\newtheorem{theorem}{Theorem}[section]
\newtheorem{lemma}[theorem]{Lemma}
\newtheorem{proposition}[theorem]{Proposition}

\newtheorem{defin}[theorem]{Definition}

\newenvironment{proof}{\noindent \textbf{Proof: }}{\hfill
$\Box$  \vspace{1ex}}
\newenvironment{definition}{\begin{defin}\em}{\end{defin}}
\newtheorem{defins}[theorem]{Definitions}

\newtheorem{exs}[theorem]{Examples}

\newtheorem{ex}[theorem]{Example}

\newtheorem{rem}[theorem]{Remark}
\newenvironment{remark}{\begin{rem}\em}{\end{rem}}
\newtheorem{rems}[theorem]{Remarks}

\newtheorem{corollary}[theorem]{Corollary}

\def\fq{\mathbb{F}_q}

\def\B{\mathbf{B}}
\def\lm{\textrm{LM}}
\def\lt{\textrm{LT}}
\def\C{\mathcal{C}}
\def\fqt{\mathbb{F}_q[\boldsymbol{t}]}
\def\L{\mathcal{L}}
\def\M{\mathcal{M}}
\def\e{\boldsymbol{e}}

\begin{document}

\begin{center}
{\Large\textbf{Next-to-minimal weight of toric codes defined over 
hypersimplices}}
\end{center}
\vspace{3ex}

\noindent\begin{center} 
\textsc{C\'{\i}cero  Carvalho\footnote[1]{Instituto de Matem\'atica e 
Estat\'{\i}stica - UFU - Brazil -- cicero@ufu.br} and Nupur 
Patanker\footnote[2]{Indian Institute of Science Bangalore - India -- 
nupurp@iisc.ac.in \\ \hspace*{3.5ex} \textbf{To appear in the Journal of 
Algebra 
and Its 
Applications.}}}
\end{center}
\vspace{3ex}

\begin{center}
\em{To Sudhir Ghorpade, on the occasion of his 60th birthday.}
\end{center}

\vspace{4ex}
\noindent
\textbf{Abstract.}
Toric codes are a type of evaluation codes introduced by J.P.\ Hansen in 2000. 
They are produced by evaluating (a vector space composed by) polynomials at the 
points of 
$(\fq^*)^s$,  the monomials of these polynomials being related to a certain 
polytope. Toric codes related to hypersimplices are the result of the 
evaluation of a vector space of square-free homogeneous polynomials of degree 
$d$.
The dimension and minimum distance of toric codes related to hypersimplices 
have been determined by 
Jaramillo et al.\ in 2021. 
The next-to-minimal weight in the case $d = 1$ has been determined by 
Jaramillo-Velez et al.\ in 2023.
In this work we use tools from Gr\"obner basis theory to determine the 
next-to-minimal weight of these codes for $d$ such that $3 \leq d \leq 
\frac{s - 2}{2}$ or $\frac{s + 2}{2} \leq d < s$.
\vspace{3ex}


\vspace{1ex}

\vspace{3ex}
\section{Introduction}

Let $\fq$ be a finite field with $q$ elements. A (linear) code $\C$ of length 
$m$ is an $\fq$-vector subspace of $\fq^m$. Other important parameters of $\C$ 
are its dimension and its minimum distance
\[
\delta(\C) := \min \{ \omega(v) \mid v \in \C, v \neq 0\}
\]
where $\omega(v)$ is the Hamming weight of $v$, namely the number of nonzero 
entries of 
$v$. 

In this paper we will work with an instance of the so-called evaluation codes. 
Let $X \subset \fq^s$, and let $I_X \subset \fq[t_1, \ldots, t_s] =: \fqt$ be 
the set 
of all polynomials which vanish at all points of $X$. For any $P \in X$ one 
may find a polynomial $f \in \fqt$ such that $f(P) = 1$ and $f(Q) = 0$ for all 
$Q \neq P$, hence it is not difficult to check that, writing $n := | X |$,  the 
evaluation map
\begin{equation}\label{varphi}
\begin{array}{rcl}
\varphi : \fqt/I_X & \longrightarrow & \fq^n \\
f + I_X & \longmapsto & (f(P_1), \ldots, f(P_n))
\end{array}
\end{equation}
is an isomorphism of $\fq$-vector spaces.
Given a subspace $\L \subset \fqt/I_X$ the image $\varphi(\L)$ is the 
evaluation code associated to $X$ and $\L$. Fitzgerald and Lax (see \cite[Prop. 
1]{fl}) proved that every linear code may be realized as an evaluation code, 
for 
appropriate $X$ and $\L$.

In this work we take $X := (\mathbb{F}_q^*)^s$, so that $I_X = (t_1^{q - 1} - 
1,\ldots, t_s^{q - 1} - 1)$ and $n = (q-1)^s$.  Let $R \subset \mathbb{R}^s$ be 
a 
lattice polytope, 
i.e.\ the convex hull of a finite set of points with integral coordinates, and 
assume that $R \subset [0, q-2]^s$. Let $\L \subset \fqt/I_X$ be the 
$\fq$-vector subspace generated by $\{ t_1^{\alpha_1} \cdots t_s^{\alpha_s} 
+ I_X 
\mid (\alpha_1, \ldots, \alpha_s) \in R \cap \mathbb{Z}^s \}$.
The toric code associated to $X$ and $R$ is $\varphi(\L)$ (see \cite{hansen}).
We study the case where $R$ is a hypersimplex of $\mathbb{R}^s$.

\begin{definition}
Let $d$ be a positive integer such that $d \leq s$,  and let $R(d)$ be the 
convex hull of the set 
$\{\e_{i_1} + \cdots + \e_{i_d} \mid 1 \leq i_1 < \cdots < i_d \leq s \}$, 
where $\e_i$ denotes the $i$-th vector in the canonical basis for 
$\mathbb{R}^s$, $1 \leq i \leq s$. 
Let $\L(d) \subset \fqt/I_X$ be the 
$\fq$-vector subspace generated by $\{ t_1^{\alpha_1} \cdots t_s^{\alpha_s} 
+ I_X 
\mid (\alpha_1, \ldots, \alpha_s) \in R(d) \cap \mathbb{Z}^s \}$. The toric 
code $\varphi(\L(d))$ associated to $X$ and $R(d)$ will be denoted by $\C(d)$.
\end{definition}

In the case of a hypersimplex $R(d)$ we have 
\[
R(d) \cap \mathbb{Z}^s  = \{\e_{i_1} + \cdots + \e_{i_d} \mid 1 \leq i_1 < 
\cdots < i_d \leq s \}
\] 
so $\L(d)$ is  $\fq$-vector space generated by the classes 
$X_1^{\alpha_1} \cdots X_s^{\alpha_s} + I_X$ where $\alpha_i \in \{0,1\}$ 
for all $i = 1, \ldots, s$ and $\sum_{i = 1}^s \alpha_i = d$.  The minimum 
distance of $\C(d)$ was determined in \cite{evalcodes}. 
The second least Hamming weight, 
also known as next-to-minimal 
weight, of $\C(d)$
was determined for $d = 1$ in \cite{jaramillo2023}.
In this work we 
determine the next-to-minimal 
weight of $\C(d)$ for $d$ such that 
$3 \leq d \leq \frac{s -2}{2}$ or $\frac{s+2}{2}  
\leq d < s$. We also characterize the 
classes of polynomials whose evaluation produces the minimum weight vectors of 
$\C(d)$.

The paper is organized as follows. In the next section we present several 
results from Gr\"obner bases theory which will be used in Sections 3 and 4.  
Since the last two sections contain many calculations, we start Section 3 with 
a description of the procedure we used to determine the next-to-minimal weight 
of $\C(d)$, so the reader may follow the ensuing calculations knowing why they 
are needed and where they will be used. Then we present preliminary  results
which will be used in Section 4, together with a characterization of 
polynomials that produce the minimum weight codewords. The last section brings 
all the previous results together to determine the value of the next-to-minimal 
weights of $\C(d)$, when $3 \leq d \leq \frac{s -2}{2}$ or $\frac{s+2}{2}  
\leq d < s$.

\section{Tools from Gr\"obner bases theory}

Let $\M$ be the set of monomials in the ring $\fqt$, and let $\prec$ be a 
monomial order in $\M$. Let $f \in \fqt$, $f \neq 0$, 
the greatest monomial which appears in $f$ is called the leading monomial of 
$f$ and is denoted by $\lm(f)$. 	 

\begin{definition}
Let $J \subset \fqt$ be an ideal. A set $\{g_1, \ldots, g_k\} \subset J$ is a 
Gr\"obner basis for $J$ w.r.t. $\prec$ if the leading monomial of any nonzero  
$f \in J$ is a multiple of $\lm(g_i)$ for some $i \in \{1, \ldots, k\}$.
\end{definition}

This concept was introduced by Bruno Buchberger in his Ph.D. thesis (see 
\cite{bruno}). 
There he proved that every nonzero ideal $J$ admits a Gr\"obner basis w.r.t.\ a 
fixed monomial order, and presented an algorithm that produces a Gr\"obner 
basis for $J$ starting from a generating set. In what follows we assume that 
the reader is familiar with Buchberger's algorithm, including the concepts of 
$S$-polynomial of $f, g \in \fqt$, 
denoted by $S(f,g)$ and of division of a polynomial by a set of polynomials 
(see e.g.\ \cite{gb-in-coding} for a concise 
presentation of these concepts, or e.g.\ \cite{cox} for a full treatment).

\begin{definition}
Let $J \subset \fqt$ be a nonzero ideal. The footprint of $J$ (w.r.t. $\prec$) 
is the set
\[
\Delta(J) := \{ M \in \M \mid M \textrm{ is not the leading monomial of any 
polynomial in  } J \}.
\]
\end{definition}

From the definitions above one may easily prove the following result.

\begin{lemma} If $\{g_1, \ldots, g_k\}$ \label{gb-fp}
is a Gr\"obner basis for $J$ then 
\[
\Delta(J) = \{ M \in \M \mid M \textrm{ is not a multiple of } \lm(g_i), i = 
1,\ldots k\}.
\]
\end{lemma}

An important result in Buchberger's thesis states that for a nonzero ideal $J 
\subset \fqt$ the set 
\[
\{ M + J \mid M \in \Delta(J) \}
\]
is a basis for $\fqt/J$ as an $\fq$-vector space.
Thus,  denoting by $\langle \Delta(J) \rangle$ the $\fq$-vector subspace of 
$\fqt$ generated by the monomials in $\Delta(J)$, for each $f \in \fqt$ there 
exists a unique $f^* \in \langle \Delta(J) \rangle$ such that $f + J = f^* + J$.

This result may be used to study the parameters of evaluation codes.
We start by noting that since the map \eqref{varphi} is an isomorphism we get 
$|\Delta(I_X)| = \dim_{\fq}(\fqt/I_X) = n = |X|$. 

\begin{proposition}\label{bound-weight}
Let $X \subset \fq^s$, and let $I_X \subset \fqt$ be 
the set 
of all polynomials which vanish at all points of $X$. Let $f \in \fqt$, $f 
\notin I_X$, 
then the weight of $\varphi(f + I_X)$ satisfies 
\[
\omega(\varphi(f + I_X)) \geq 
| \Delta(I_X) | -  |\Delta(I_X + (f))|. 
\]
\end{proposition}
\begin{proof}
If $f(Q) \neq 0$ for all $Q \in X$ then $\omega(\varphi(f + I_X)) = |X| = | 
\Delta(I_X)|$ and the inequality holds trivially. So we assume that there exist 
a total of 
$N 
\geq 1$ points of $X$, namely 
$\{Q_1, \ldots, Q_{N}\}$, which are zeros of $f$.
The evaluation map $\psi: \fqt/(I_X + (f)) \rightarrow \fq^{N}$
defined by $\psi(g + (I_X + (f))) = (g(Q_1), \ldots, g(Q_N))$, 
is a linear 
transformation which is surjective (for the same reason that $\varphi$ is 
surjective), so from Buchberger's thesis result we get $N \leq |\Delta(I_X + 
(f))|$.
Hence
\[
\omega(\varphi(f + I_X)) = n - N \geq 
 | \Delta(I_X) | -  |\Delta(I_X + (f))|.
\]  
%
%
\end{proof}

The results mentioned above have been used to determine the minimum distance,  
as well as the next-to-minimal weights, of several types of codes (see e.g.
\cite{geil}, \cite{geil2}, \cite{rolland}, \cite{car-2013}, 
\cite{car-neu-2017}, \cite{car-2024} and the references therein). In the next 
sections we will use them to determine the next-to-minimal weights of 
toric codes defined over  hypersimplices.

\section{Preliminary results}

From now on we work with the graded lexicographic order on the set $\M$ of 
monomials of $\fqt$, where $t_s \prec \cdots \prec t_1$. 

Let $X := (\mathbb{F}_q^*)^s$, so that $I_X = (t_1^{q - 1} - 
1,\ldots, t_s^{q - 1} - 1)$. Since the leading monomials of any two distinct 
generators of $I_X$ are coprime, we get (see \cite[p. 103--104]{cox}) that 
$\{t_1^{q - 1} - 1,\ldots, t_s^{q - 1} - 1\}$ is a Gr\"obner basis for $I_X$, 
and from Lemma \ref{gb-fp} we get that
\[
\Delta(I_X) = \left\{ \prod_{i = 1}^s t_i^{\alpha_i} \in \M \mid 0 \leq 
\alpha_i 
\leq 
q - 2 \; \forall \; i = 1, \ldots, s\right\}.
\]

Let $d$ be an integer such that $3 \leq d < s$ and assume that $q \geq 4$. Let 
$\L(d)$ be the $\fq$-vector subspace of $\fqt/I_X$ generated by 
\[
\{ t_1^{\alpha_1} \cdots t_s^{\alpha_s} + I_X \mid \sum_{i = 1}^s \alpha_i 
= d, \, \alpha_i \in \{0, 1\} \; \forall \; i = 1, \ldots, s\},
\]
and let $\C(d) = \varphi(\L(d))$. From the fact that $\varphi$ is an 
isomorphism and that the monomials whose classes generate $\L(d)$ are in 
$\Delta(I_X)$ we get that $\dim_{\fq}\L(d) = \binom{s}{d}$. The dimension and 
the minimum 
distance $\delta(\C(d))$ were determined in \cite{evalcodes}:
\[
\delta(\C(d)) = \left\{ \begin{array}{l} (q - 2)^d (q - 1)^{s - d} \textrm{ if 
} 3 \leq d \leq \frac{s}{2} \; ;\\
                                       (q - 2)^{s -d} (q - 1)^{d} \textrm{ if } 
                                       \frac{s}{2} < d < s .
                         \end{array} 
                 \right.        
\]

Before we start a series of computations that will determine the 
next-to-minimal weight of $\C(d)$ for most values of $d$, we describe the 
general idea of the procedures that we will use. Let $f$ be a nonzero 
polynomial which
is an $\fq$-linear combination of monomials in the set
\[
\left\{ t_1^{\alpha_1} \cdots t_s^{\alpha_s}  \mid \sum_{i = 1}^s \alpha_i 
= d, \, \alpha_i \in \{0, 1\} \; \forall \; i = 1, \ldots, s \right\},
\]
and let 
$\B := \{t_i^{q - 1} - 1 \mid i = 1, \ldots, s\} \cup \{ 
f\}$. In what follows, we will show that, in the case where $2d \leq s$,  if 
$\B$ is a Gr\"obner basis for the ideal $I_X + (f)$ it defines, then 
$\varphi(f  + I_X)$ is 
a minimum weight codeword (see Corollary \ref{rem-eq-zero} and Proposition 
\ref{min-word}). 
Thus, for $\varphi(f  + I_X)$ to have a weight greater than the minimum 
distance, the $S$-polynomial $S(t_j^{q - 1} - 1, f)$, for some $j \in 
\{1,\ldots, s\}$, must have a nonzero remainder $z$ in the division by $\B$. 
In the first Theorem of this section we prove that the leading monomial of $z$
can be of four distinct types. 
From the division algorithm we know that $\lm(z)$ is not a multiple of $t_i^{q 
- 1}$, for all $i = 1, \ldots, s$ (i.e. $\lm(z) \in \Delta(I_X)$) and also not 
a 
multiple of $\lm(f)$.
Since $\B \cup \{z\}$ is (also) a generating set for $I_X + (f)$, from the 
definition of footprint we get that
\begin{equation*} 
\begin{split}
\Delta(I_X + (f)) &\subset \{ M \in \M \mid  t_i^{q - 1} \nmid M \; \forall \; 
i 
= 1, \ldots, s; \lm(f)\nmid M , \lm(z) \nmid M \}  \\
&= \Delta(I_X) \setminus \big(\{M \in \M \mid M \textrm{ is a multiple of } 
\lm(f) 
\} \\
& \textrm{\hspace{20ex}} \cup \{M \in \M \mid 
M \textrm{ is a multiple of }  \lm(z)\}\big).
\end{split}
\end{equation*}
Thus 
\begin{equation*} 
\begin{split}
| \Delta&(I_X + (f)) | \leq |\Delta(I_X)| - |\{M \in \Delta(I_X) \mid  
 M \textrm{ is a multiple of } \lm(f)  \}| \\
& -  |\{M \in \Delta(I_X) \mid M 
\textrm{ is a multiple of } \lm(z) \textrm{ and not a multiple of } \lm(f) \}|.
\end{split}
\end{equation*}
 Note that $\lm(f) = t_{i_1} \cdots t_{i_d}$ for some $1 \leq i_1 < \cdots 
 < i_d \leq s$ so 
\[
|\{M \in \Delta(I_X) \mid  M \textrm{ is a multiple of } \lm(f)  \}| = 
(q - 2)^d (q - 1)^{s - d}
\] 
and from Lemma \ref{bound-weight} we get 
\begin{equation}\label{bound-2nd}
 \begin{split}  
\omega(\varphi(f + I_X)) &\geq 
| \Delta(I_X) | - | \Delta(I_X + (f))| \\
&\geq  (q - 2)^d (q - 1)^{s - d} \\
&+ |\{M \in \Delta(I_X) \mid  M 
\textrm{ is a multiple of } \lm(z) \\
&\textrm{\hspace{25ex}} \textrm{ and not a multiple of } \lm(f) \}|.
\end{split}
\end{equation}
In the next section, for each of the four types of $\lm(z)$ we determine 
the number 
\begin{equation*} 
\begin{split}
N(\lm(z)) := 
|\{M \in \Delta(I_X) 	&\mid  M 
\textrm{ is a multiple of } \lm(z) \\
 &\textrm{\hspace{10ex}} \textrm{ and not a multiple of } \lm(f) \}|
\end{split}
\end{equation*}
(see Lemma \ref{num_mon}). 
We prove that the lowest possible value is 
$(q - 3)(q - 2)^{d - 2}(q - 1)^{s - d - 1}$, and the second lowest is
$(q - 2)^{d}(q - 1)^{s - d - 2}$ (see Proposition \ref{nm4-minimal}).
Thus, when $z$ has 
a leading monomial (say $M^*$)
of the type that has the lowest value for  $N(\lm(z))$
 we have
\[
\omega(\varphi(f + I_X)) \geq
(q - 2)^d (q - 1)^{s - d} + (q - 3)(q - 2)^{d - 2}(q - 1)^{s - d - 1}.
\]
Yet this lower bound is not attained because we 
prove (see Proposition \ref{m3}) that, when $\lm(z) = M^*$, 
there exists  $z' \in I_X + (f)$ with a leading monomial $M^{**}$ which is 
in 
$\Delta(I_X)$ and is not a multiple of $\lm(f)$ and not a multiple of $M^*$. 
Thus $\B \cup \{z, z'\}$ is a generating set for $I_X + (f)$ and reasoning as 
above  we get	 
\begin{equation*} 
\begin{split}
\omega(\varphi(f + I_X)) &\geq 
| \Delta(I_X) \setminus \Delta(I_X + (f))| = 
| \Delta(I_X) | - | \Delta(I_X + (f))| \\
&\geq  (q - 2)^d (q - 1)^{s - d} + (q - 3)(q - 2)^{d - 2}(q - 1)^{s - d - 1}\\
&+ |\{M \in \Delta(I_X) \mid  M 
\textrm{ is a multiple of } M^{**} \\
&\textrm{\hspace{3ex}} \textrm{ and not a multiple of } \lm(f) \textrm{ and not 
a multiple of } M^{*}\}|.
\end{split}
\end{equation*}
We also show in Proposition \ref{m3} that the above lower bound is greater than 
$(q - 2)^d (q - 1)^{s - d} + (q - 2)^{d}(q - 1)^{s - d - 2}$, a lower bound 
that we prove is realized, in the case $d \leq 
\frac{s -2}{2}$, when $\lm(z)$ is of the type for which $N(\lm(z))$ 
has the second lowest possible value (see Theorem \ref{main}). This concludes 
the determination of the next-to-minimal weight in the case where $d \leq 
\frac{s-2}{2}$, and in the last result we show how to obtain the 
next-to-minimal weight in the case where $\frac{s+2}{2}  \leq  d < s$.

Now we start the computations described in the above procedure.
In this paper polynomials whose monomials are square-free are called square-free
polynomials.
Recall that we are assuming $q \geq 4$ and $d \geq 3$, and  
recall also that $f$ is a 
$\fq$-linear combination of monomials of the type 
$t_1^{\alpha_1} \cdots t_s^{\alpha_s}$
with $\alpha_i \in \{0, 1\}$ for all $i = 1, \ldots, s$
and $\sum_{i = 1}^s \alpha_i = d$, or more simply,
a homogeneous square-free polynomial of degree $d$ which we may assume is monic.
To simplify the 
notation, after a  relabeling  
of the variables we assume, from now on, that $\lm(f) = t_1\cdots t_d$. 
Firstly we 
determine the possibilities of leading monomial for the remainder in the 
division of $S(t_j^{q - 1} - 1, f)$ by the set $\B$, for $j \in \{1, \ldots, 
d\}$. 
As described 
above, we always work with monomials which belong to the footprint of $I_X$. 
So, 
in what follows, when we say we are going to count the number of monomials 
which are multiple (or are not multiple) of a monomial $M$, it's to be 
understood that we are counting the number of monomials in $\Delta(I_X)$ 
that are multiple (or are not multiple) of $M$.
When we write $\widehat{t_i}$ for a variable in a product, we mean that the 
factor $t_i$ does not 
appear in the product, where $i \in \{1, \ldots, s\}$.

\begin{theorem}\label{rem-mon} Let $f \in \fqt$ be a
 homogeneous square-free monic polynomial of degree $d$
and let $j\in \{1, \ldots, d\}$. The 
remainder of 
$S(t_j^{q - 1} - 1, 
f)$ in the division by $\B = \{t_i^{q - 1} - 1 \mid i = 1, \ldots, s\} \cup \{ 
f\}$ 
is either zero, or has as leading monomial:\\
a) the monomial $t_1 \cdots \widehat{t_j} \cdots
t_d$; \\
b) a monomial of degree $q - 2 + d$, of the form $t_{j}^{q - 2} M_d$ or of the 
form $t_{j}^{q - 2} t_{e}^2  M_{d -2}$, with $e \in \{d+1, \ldots, s\}$, $t_e 
\nmid M_{d - 2}$, and 
where for $i = d-2, d$ the monomial 
$M_i$ is  square-free 
of degree $i$, and is not a multiple of either $t_j$ or $t_{\ell}$, for some 
$\ell 
\in \{1, \ldots, 
d\} \setminus \{j\}$; \\
c) a monomial of degree $q - 2 + d$ of the form $t_1 \cdots \widehat{t_j}  
\cdots  t_d t_{e_1}^{q - 2} t_{e_2}$, where $e_1$ and $e_2$ are distinct and 
are 
in the set $\{d+1, \ldots, s\}$.
\end{theorem}
\begin{proof}
Let $j \in \{1, \ldots, d\}$, we write $f$ as  
\begin{equation}\label{f}
f =  t_1 \cdots t_d + t_j h_{d - 1} + t_1 \cdots \widehat{t_j}  \cdots 
 t_d h_1 + h_d 
\end{equation}
where for all $i = 1, d-1, d$ we have that $h_i$, if not equal to zero, is a 
homogeneous square-free polynomial of degree $i$ such that 
every monomial is not a multiple of $t_j$ and also does not contain all the 
variables in $\{t_1,\ldots, t_d\} \setminus \{t_j\}$ (in particular, if $h_1 
\neq 0$ 
then $h_1$ is a linear form in the variables 
$\{t_{d + 1}, \ldots, t_s\}$, since $f$ is square-free).

With this, we have
\begin{equation*}
\begin{split}
S(t_j^{q - 1} - 1,& f) = \,t_1 \cdots \widehat{t_j}  \cdots 
 t_d (t_j^{q - 1} - 1) - t_j^{q - 2} f \\
=&  - t_1 \cdots \widehat{t_j}  \cdots t_d - t_j^{q - 1} h_{d - 1} - 
(t_1 \cdots t_j^{q - 2} \cdots t_d) h_1 - t_j^{q - 2} h_d.
\end{split}
\end{equation*}

If $h_{d - 1} \neq 0$ then 
dividing $S(t_j^{q - 1} - 1, f)$ by $t_j^{q - 1} - 1$ we get $-h_{d - 1}$ for 
quotient and the remainder is 
\begin{equation*}
\begin{split}
S_1 :=&\, S(t_j^{q - 1} - 1, f) + h_{d - 1}(t_j^{q 
- 1} - 1) \\ =& - t_1 \cdots \widehat{t_j} \cdots t_d -  h_{d - 1} -  
(t_1 \cdots t_j^{q - 2} \cdots t_d) h_1 - t_j^{q - 2} h_d. 
\end{split}
\end{equation*}
If $h_{d - 1} = 0$ 
then $S_1 =  S(t_j^{q - 1} - 1, f)$, so we continue with $S_1$, regardless of 
$h_{d-1}$ being zero or not. 

If $h_1 = 0$ then no monomial in $S_1$ is a multiple of $t_i^{q - 1}$, for 
$i = 
1,\ldots, d$ or multiple of $\lm(f)$, so $S_1$ is the remainder in the division 
of 
$S(t_j^{q - 1} - 1, f)$ by $\B$. To analyze the  possibilities for the 
leading monomial of 
$S_1$, we consider two cases: \\
a) if  $h_d = 0$ then the leading monomial of $S_1$ is $t_1 \cdots  
\widehat{t_j}  \cdots t_d$, because all 
monomials of $h_{d - 1}$ (in case $h_{d - 1} \neq 0$) do not have the variable 
$t_j$ and certainly have at least one of the variables $t_{d+1}, \ldots, t_s$; 
\\
b) if  $h_d \neq  0$ then the leading monomial of $S_1$ is a monomial of the 
type $t_j^{q - 2} M$, where $M$ is a  square-free monomial 
of degree $d$ in $\{t_{1}, \ldots, t_{s}\}\setminus \{t_j\}$,
which does not contain all variables in $\{t_1,  \ldots , 
t_d\}\setminus\{t_j\}$. 
This is because
 $t_j^{q - 2} h_d$ is the homogeneous part of $S_1$ with the highest degree, 
 which is $q - 2 + d$. This proves the theorem 
in the case where $h_1 = 0$.

From now on we assume that $h_1 \neq 0$ (we recall that we are also assuming $q 
\geq 4$).  
We will do a close examination of the division of $S_1$ by $\B$, and for 
convenience we write  $h_1 = \sum_{i = 1}^u a_i M_i$, with $a_i \in 
\mathbb{F}_q^*$ 
for all $i 
= 
1, \ldots, u$, and $M_u \prec \cdots \prec M_1$, so that 
\[
S_1 = - (t_1\cdots  t_d)t_j^{\alpha}\left( \sum_{i = 1}^u a_i M_i \right) + 
r, 
\]
where $\alpha = q - 3$ and $r = - t_1  \cdots  \widehat{t_j}  \cdots  t_d 
-  
h_{d - 1} - t_j^{q - 2} h_d$.
We note that no monomial of $S_1$ is a multiple of $t_i^{q - 1}$, 
for $i = 1, \ldots, s$, and the monomials in
$-(t_1 \cdots  t_d)t_j^{\alpha}( \sum_{i = 1}^u a_i M_i )$ are the only ones 
in $S_1$ 
which are 
multiple of $\lm(f)$, so the division of $S_1$ by $\B$ starts with the 
division of $S_1$ by $f$.
We have $\lt(-(t_1 \cdots  t_d)t_j^{\alpha}( \sum_{i = 1}^u a_i M_i )) = 
- a_1(t_1 \cdots  
t_d)t_j^{\alpha} M_1$, so the first term of the quotient in the division is  
$- a_1 t_j^{\alpha} M_1$, and the first partial remainder is 

\begin{equation*}
\begin{split}
r_1 &= - (t_1 \cdots  t_d)t_j^{\alpha}\left( \sum_{i = 2}^u a_i M_i \right) + 
r \\
&\textrm{\hspace{22ex}} + a_1 t_j^{\alpha} M_1 (t_j h_{d - 1} + t_1  \cdots  
\widehat{t_j}  \cdots 
t_d h_1 + h_d).
\end{split}
\end{equation*}

As above, no monomial in $r_1$ is a multiple of $t_i^{q - 1}$, 
for $i = 1, \ldots, s$, and the monomials in the homogeneous polynomial 
\[
\tilde{r}_1 := - (t_1 \cdots  t_d)t_j^{\alpha}( \sum_{i = 2}^u a_i M_i ) + 
a_1 
t_j^{\alpha} 
M_1 t_1  \cdots  \widehat{t_j}  \cdots 
t_d h_1
\]
are the only ones 
in $r_1$ 
which are 
multiple of $\lm(f)$, so we proceed with the division by dividing $r_1$ by $f$.
We have $\lt(\tilde{r}_1) = 
- a_2(t_1 \cdots  
t_d)t_j^{\alpha} M_2$, so we add as 
second summand 
in the quotient the term $- a_2 t_j^{\alpha} M_2$, and the second 
partial remainder is 
\begin{equation*}
\begin{split}
r_2 &= S_1 + (t_j^\alpha(a_1M_1 + a_2 M_2)) f  
\\ 
&= r_1 + t_j^\alpha a_2 M_2  f \\ 
&= - (t_1 \cdots  t_d)t_j^{\alpha}\left( \sum_{i = 3}^u a_i M_i \right) + r \\
& + 
a_1 t_j^{\alpha} M_1 (t_j h_{d - 1} + t_1  \cdots  \widehat{t_j}  \cdots 
t_d h_1 + h_d)\\
& + 
a_2 t_j^{\alpha} M_2 (t_j h_{d - 1} + t_1  \cdots  \widehat{t_j}  \cdots 
t_d h_1 + h_d)
\end{split}
\end{equation*}

%
%
%
%

Again, no monomial in $r_2$ is a multiple of $t_i^{q - 1}$, 
for $i = 1, \ldots, s$, and the monomials in the homogeneous polynomial 
\[
\tilde{r}_2 := - (t_1 \cdots  t_d)t_j^{\alpha}\left( \sum_{i = 3}^u a_i M_i 
\right) 
+ t_j^{\alpha} (a_1 M_1 + a_2 M_2)  t_1  \cdots  \widehat{t_j}  \cdots 
t_d h_1
\]
are the only ones in $r_2$ which are 
multiple of $\lm(f)$.
Thus, after $u$ steps, we 
get the partial remainder
\begin{equation*}
\begin{split}
S_2 &:= S_1 + t_j^{q - 3} h_1 f \\
&= - t_1  \cdots  \widehat{t_j}  \cdots  t_d -  h_{d - 1} -  
 t_j^{q - 2} h_d\\
& + t_j^{q - 2}h_1 h_{d - 1} + t_1  \cdots  
t_j^{q - 3}  \cdots t_d h_1^2 + t_j^{q - 3} h_1 h_d  
\end{split}
\end{equation*}

Since $q - 3 > 0$ the monomials in
$(t_1  \cdots  t_j^{q - 3}  \cdots  t_d) h_1^2$ are the only ones in $S_2$ 
which are 
multiple of $\lm(f)$, and no monomial of $S_2$ is a multiple of $t_i^{q - 1}$, 
for $i = 1, \ldots, s$, so we continue with dividing $S_2$ by $f$. We write  
$h_1^2 = \sum_{i = 1}^{\tilde{u}} \tilde{a}_i \tilde{M}_i$, with $\tilde{a}_i 
\in \mathbb{F}_q^*$ for all $i 
= 1, \ldots, \tilde{u}$, and $\tilde{M}_1 \succ \cdots \succ 
\tilde{M}_{\tilde{u}}$, 
and we also write
\[
S_2 = (t_1 \cdots  t_d)t_j^{\tilde{\alpha}}\left( \sum_{i = 1}^{\tilde{u}} 
\tilde{a}_i \tilde{M}_i \right) + \tilde{r}, 
\]
where $\tilde{\alpha} = q - 4$ and 
\begin{equation*}
\begin{split}
\tilde{r} &= - t_1  \cdots  \widehat{t_j}  \cdots  t_d -  h_{d - 1} -  
 t_j^{q - 2} h_d\\
& + t_j^{q - 2}h_1 h_{d - 1}  + t_j^{q - 3} h_1 h_d.  
\end{split}
\end{equation*}

We are in the same situation as before, because all the monomials in $h_1^2$, 
$h_1 h_{d - 1}$, and $h_1 h_d$ are not multiples of $t_j$, and each of them is 
also not a multiple of at least one $t_i \in \{t_1, \ldots, t_d\} \setminus 
\{t_j\}$. The first term of the quotient will be $t^{\alpha}_j \tilde{a}_1 
\tilde{M}_1$, and proceeding like this,
after $\tilde{u}$ steps the partial remainder in the division of $S_2$ by 
$f$ is 
\begin{equation*}
\begin{split}
S_3 &:= S_2 - t_j^{q - 4} h_1^2 f   \\
&= - t_1  \cdots  \widehat{t_j}  \cdots  t_d -  h_{d - 1} -  
  t_j^{q - 2} h_d
  + t_j^{q - 2}h_1 h_{d - 1} + t_j^{q - 3} h_1 h_d\\ 
& \textrm{\hspace{2cm}} - t_j^{q - 3} h_1^2 h_{d - 1} - t_1  \cdots  
t_j^{q - 4}  \cdots 
t_d h_1^3 - t_j^{q - 4} h_1^2 h_d  
\end{split}
\end{equation*}

Since $S_2 = S_1 + t_j^{q - 3} h_1 f$, we have 
\[
S_3 = S_2 - t_j^{q - 4} h_1^2 
f  =  S_1 + t_j^{q - 3} h_1 f - t_j^{q - 4} h_1^2 f.
\]

Thus, repeating the steps above, we arrive at a partial remainder
\[
S_{q - 1} = S_1 + \left( (-1)^2 t_j^{q - 3}h_1 + (-1)^3 t_j^{q - 4} h_1^2 + 
\cdots 
+ (-1)^{q - 1} h_1^{q - 2}\right) f.
\]
Writing $(-1)^{q - 1} = 1$, which is true regardless of $q$ being odd or 
even, and replacing $S_1$ and $f$ by their expressions, we get 
\begin{equation}  \label{eq1}
\begin{split}
S_{q - 1} &=  t_1  \cdots  \widehat{t_j}  \cdots  t_d (h_1^{q -1} - 1) \\
&+ h_{d - 1}(-1 + t_j^{q - 2} h_1 - \dots + t_j h_1^{q -2}) \\
&+ h_d(- t_j^{q - 2} + t_j^{q - 3} h_1 - \cdots + h_1^{q - 2} ).
\end{split}
\end{equation}
No monomial of $S_{q - 1}$ is a multiple of $\lm(f)$. To continue the division
of $S(t_j^{q - 1} - 1, f)$ by $\B$  we consider 
two cases. \\
I) In the first case we assume that $h_1 = a t_e$, with $e \in \{d+1, \ldots, 
s\}$ and $a \in \fq^*$. Then $S_{q - 1}$ may be divided by $t_e^{q - 1} - 1$, 
the quotient is 
$ t_1  \cdots  \widehat{t_j}  \cdots  t_d $ (recall that $a^{q - 1}= 1$) 
and 
we get as remainder

\begin{equation} \label{eq2}
\begin{split}
\tilde{S}_{q - 1} &= S_{q - 1} - (t_1  \cdots  \widehat{t_j}  \cdots 
t_d)(t_e^{q - 1} - 1) \\
&= h_{d - 1}(- 1 + t_j^{q - 2} a t_e - \cdots + t_j a^{q - 2} t_e^{q -2}) \\
&+ h_d(-t_j^{q - 2} + t_j^{q - 3} a t_e - \cdots + a^{q - 2}t_e^{q - 2} ) \\
&= - h_{d - 1}  +  t_j^{q - 2}(a t_e h_{d - 1}- h_d)  
- t_j^{q - 3} a t_e (a t_e h_{d - 1} - h_d) \\
&\textrm{\hspace{3cm}}+ \cdots + 
t_j a^{q - 3} t_e^{q - 3}(a t_e h_{d - 1} - h_d) + h_d a^{q - 2} t_e^{q -2}.
\end{split}
\end{equation}

If $a t_e h_{d - 1} - h_d \neq 0$ then all monomials in $\tilde{S}_{q - 1}$ are 
not multiples of $\lm(f)$ or any $t_i^{q - 1}$, for $i = 1, \ldots, s$, thus  
$\tilde{S}_{q - 1}$ is the (final) remainder in the division of $S(t_j^{q - 1} 
- 
1, f)$ 
by $\B$. The leading monomial of $\tilde{S}_{q - 1}$ is $t_j^{q - 2} \lm(a t_e 
h_{d - 1}- h_d)$ which is of the form $t_{j}^{q - 2} t_{i_1}  
\cdots  t_{i_d}$ or of the form $t_{j}^{q - 2} t_{e}^2  t_{i_2} \cdots  
t_{i_{d-1}}$ (the latter happens when $\lm(t_e h_{d - 1}) > \lm(h_d)$ and $t_e$ 
is a divisor of $\lm(h_{d - 1})$). This proves the theorem when $h_1 = a t_e$ 
and $a t_e h_{d - 1} - h_d \neq 0$.

Assume now that $a t_e h_{d - 1} = h_d$. From equation \eqref{eq2}  we get that 
$\tilde{S}_{q - 1} = - h_{d - 1} + h_d a^{q - 2} t_e^{q - 2} = 
h_{d - 1}(t_e^{q - 1} - 1)$, so the final remainder 
in the division of $S(t_j^{q - 1} - 1, f)$  by $\B$ is zero, and the theorem is 
proved in this case.
\vspace{1ex}

\noindent
II) In this last case we assume that $h_1 = \sum_{i = 1}^u a_i t_{e_i}$, where 
$a_i \in \mathbb{F}_q^*$ for all $i = 1, \ldots, u$, $u \geq 2$ and $t_{e_1} > 
\cdots > 
t_{e_u}$. We have 
\begin{equation} \label{h1^{q-1}}
h_1^{q - 1} = \sum_{i = 1}^u t_{e_i}^{q - 1} + (q - 1) 
a_1^{q - 2} a_2 t_{e_1}^{q - 2} t_{e_2} + \tilde{h},
\end{equation}
where  
$\tilde{h}$ is a nonzero homogeneous polynomial of degree $q - 1$ with all 
monomials less than $t_{e_1}^{q - 2} t_{e_2}$ and no monomial of the form 
$t_{e_i}^{q - 1}$, with $i = 1, \ldots, u$. We rewrite equation 
\eqref{eq1} as
\begin{equation*} 
\begin{split}
S_{q - 1} &=  t_1  \cdots  \widehat{t_j}  \cdots  t_d (\sum_{i = 1}^u 
t_{e_i}^{q - 1} - 1) \\
& + t_1  \cdots  \widehat{t_j}  \cdots  t_d
 ((q - 1) a_1^{q - 2} 
a_2 t_{e_1}^{q - 2} t_{e_2} + \tilde{h}) \\
&+ h_{d - 1}(- 1 + t_j^{q - 2} h_1 - \dots + t_j h_1^{q -2}) \\
&+ h_d(-t_j^{q - 2} + t_j^{q - 3} h_1 - \cdots + h_1^{q - 2}) ,
\end{split}
\end{equation*}
so dividing $S_{q - 1}$ successively by $t_{e_1}^{q - 1} - 1, \ldots, 
t_{e_u}^{q - 1} - 1$, we get the remainder
\begin{equation*} 
\begin{split}
\tilde{S}_{q - 1} &=  (u -1) t_1  \cdots  \widehat{t_j}  \cdots  t_d \\
& + t_1  \cdots  \widehat{t_j}  \cdots  t_d
 ((q - 1) a_1^{q - 2} 
a_2 t_{e_1}^{q - 2} t_{e_2} + \tilde{h}) \\
&+ h_{d - 1}(- 1 + t_j^{q - 2} h_1 - \dots + t_j h_1^{q -2}) \\
&+ h_d(-t_j^{q - 2} + t_j^{q - 3} h_1 - \cdots +  h_1^{q - 2}) .
\end{split}
\end{equation*}

We have
\[
\lm(t_1  \cdots  \widehat{t_j}  \cdots  t_d
 ((q - 1) a_1^{q - 2} 
a_2 t_{e_1}^{q - 2} t_{e_2} + \tilde{h})) =  t_1  \cdots  \widehat{t_j}  
\cdots  t_d t_{e_1}^{q - 2} t_{e_2}, 
\]
\[
\lm(h_{d - 1}(- 1 + t_j^{q - 2} h_1 - \dots + t_j h_1^{q -2})) = t_j^{q - 2} 
t_{e_1} \lm(h_{d - 1})
\]
and
\[
\lm(h_d(-t_j^{q - 2} + t_j^{q - 3} h_1 - \cdots + h_1^{q - 2} )) =
t_j^{q - 2}  \lm(h_{d})
\]
and these are the possibilities for the leading monomial of 
$\tilde{S}_{q - 1}$. 

No monomial of $\tilde{S}_{q - 1}$ is a multiple of $\lm(f)$. The only 
possibility to appear, in $\tilde{S}_{q - 1}$, a monomial which is a multiple 
of 
$t_i^{q - 1}$ for some $i \in \{1, \ldots, s\}$, is in the products $h_{d - 1} 
t_j 
h_1^{q -2}$ and $h_{d} h_1^{q - 2}$. Indeed, if $t_{e_k}$ divides some 
monomial 
of $h_{d - 1}$ or $h_d$, for $k \in \{1, \ldots, u\}$, then the monomial 
$t_{e_k}^{q - 2}$ appears in $h_1^{q - 2}$, and we will have $t_{e_k}^{q - 1}$
appearing in a monomial $M$ in $h_{d - 1} t_j 
h_1^{q -2}$ or $h_{d} h_1^{q - 2}$. In this case we may divide $\tilde{S}_{q - 
1}$ by 
$t_{e_k}^{q - 1} - 1$ 
obtaining a remainder which coincides with $\tilde{S}_{q - 1}$, except that in 
the monomial $M$ we replace $t_{e_k}^{q - 1}$ by $1$. After we do all these 
divisions, if it's the case, we will arrive at the final remainder, and the 
possibilities  for the leading monomial are the same which were listed above. 
Observe that  $t_{e_1}$ may divide 
$\lm(h_{d - 1})$, in which case $t_j^{q - 2} 
t_{e_1} \lm(h_{d - 1}) = t_j^{q - 2} t_{e_1}^2 t_{i_2}  \cdots  t_{i_{d-1}}$. 
This 
finishes the proof of the theorem.
\end{proof}

We make explicit a consequence of the above proof which will be useful in what follows.

\begin{corollary} \label{rem-eq-zero}
Let $f \in \fqt$ be 
a homogeneous square-free monic polynomial of degree $d$ such that $\lm(f) = 
t_1  \cdots  
t_d$, and let $j \in \{1, \ldots, d\}$. Then the remainder in the division of 
$S(t_j^{q - 1} - 1, f)$ by $\{t_i^{q - 1} - 1 \mid i = 1, \ldots, s\} \cup \{ 
f\}$ is  zero if and only if $t_j + a t_{e}$ divides $f$, 
where $t_{e} \in \{t_{d + 1}, \ldots, t_s\}$ and $a \in \fq^*$.
\end{corollary}
\begin{proof}
Write $f$ as 
\begin{equation*}
f =  t_1  \cdots  t_d + t_j h_{d - 1} + t_1  \cdots  \widehat{t_j}  \cdots 
t_d h_1 + h_d 
\end{equation*}
with $h_1$,  $h_{d - 1}$ and  $h_d$ as in the above proof.
Note that the proof shows that 
the remainder is zero if and only if $h_1 = a t_e$ and $h_d = a t_e h_{d-1}$,
with $t_{e} \in \{t_{d + 1}. \ldots, t_s\}$ and $a \in \fq^*$.
And these conditions are equivalent to 
\begin{equation*} 
\begin{split}
f &=  t_1  \cdots  t_d + t_j h_{d - 1} + a t_1  \cdots  \widehat{t_j} 
 \cdots  t_d t_{e} + a t_{e} h_{d - 1} \\
&= (t_j + a t_{e})(t_1  \cdots  \widehat{t_j}  \cdots  t_d + h_{d - 
1}).
\end{split}
\end{equation*}
\end{proof}

We use the above results to characterize minimum distance codewords of 
$C(d)$.

\begin{proposition} \label{min-word}
Let $f \in \fqt$ be a homogeneous square-free monic  
polynomial of degree $d$, such that $\lm(f) = t_1  \cdots t_d$, and assume 
that 
$2d \leq s$. Then 
$\varphi(f + I_X)$ is a minimum weight codeword if and only if $f = (t_1 + 
a_{c_1} t_{c_1}) \cdots  (t_d + a_{c_d} t_{c_d})$, with $c_1, \ldots, c_d \in 
\{d+1, \ldots, s\}$ and  $a_{c_1}, \ldots, a_{c_d} \in \fq^*$.
\end{proposition}
\begin{proof}
From Proposition \ref{bound-weight} we know that 
$\omega(\varphi(f + I_X)) \geq | \Delta(I_X)| - |\Delta(I_X + (f))|$, 
and from 
\[
\Delta(I_X + (f)) \subset \Delta(I_X) \setminus \{ M \in \Delta(I_X) \mid M 
\textrm{ is a 
multiple of } \lm(f) \}
\]
we get 
$| \Delta(I_X)| - |\Delta(I_X + (f))| \geq (q - 2)^d (q - 1)^{s - d}$.
From \cite[Thm. 4.5]{evalcodes} we know that the minimum 
distance, in the case $2 d \leq s$, is exactly $(q - 2)^d(q - 1)^{s - d}$.
%
%
%
%

Assume that $f$ is such that $w(\varphi(f + I_X)) = (q - 2)^d(q - 1)^{s - d}$, 
then 
$|\Delta(I_X + (f))| = (q - 1)^s - (q - 2)^d(q - 1)^{s - d}$ and from 
Buchberger's algorithm we get that
$\{t_1^{q - 1} -1, \ldots, 
t_s^{q - 1} -1, f\}$ must be a Gr\"obner basis for $I_X + (f)$, otherwise we 
would have to add at least one new generator to
$\{t_1^{q - 1} -1, \ldots, 
t_s^{q - 1} -1, f\}$ in order to have a Gr\"obner basis, which would imply that
$|\Delta(I_X + (f))| < (q - 1)^s - (q - 2)^d(q - 1)^{s - d}$.  In particular, 
the remainder of $S(t_j^{q - 1} - 1, f)$ in the division by $\{t_i^{q - 1} 
- 1 
\mid i = 1, \ldots, s\} \cup \{ f\}$ is zero for all $j = 1, \ldots, d$. 
Hence, from Corollary \ref{rem-eq-zero} we get that
%
%
$f$ must be of the form $f = (t_1 + a_{c_1} t_{c_1}) \cdots  (t_d + a_{c_d} 
t_{c_d})$.

On the other hand, if $f = (t_1 + a_{c_1} t_{c_1}) \cdots  (t_d + a_{c_d} 
t_{c_d})$ then from all $(q - 1)^2$ possibilities  of values $(b_1,b_{c_1})$ 
for the pair of variables $(t_1, t_{c_1})$, we have that exactly $q - 1$ are 
zeros of $t_1 + a_{c_1} t_{c_1}$, namely $(b, - b (a_{c_1})^{-1})$, with $b 
\in \fq^*$. So we have $(q - 1)^2 - (q - 1) = (q -1)(q - 2)$ pair of values 
$(b_1,b_{c_1}) \in (\fq^*)^2$
which are not zeros of $t_1 + a_{c_1} t_{c_1}$. The calculation is the same 
for the other pairs $(t_j, t_{c_j})$, $j = 2, \ldots, d$, so the number of 
$2d$-tuples $(c_1, c_{c_1}, \ldots, c_d, c_{c_d})$ which are not zeros of $f$ 
is $(q -1)^d(q - 2)^d$. Thus the number of $s$-tuples in $(\fq^*)^s$ which are 
not zeros of $f$ 
is $(q -1)^d(q - 2)^d (q - 1)^{s - 2d} = (q - 2)^d(q - 1)^{s - d}$, which 
proves 
that $\varphi(f + I_X)$ is a minimum weight codeword.
\end{proof}

\section{Main results}
According to Theorem \ref{rem-mon} we have four 
possibilities for the leading monomial
of a nonzero 
remainder of $S(t_j^{q - 1} - 1, f)$ in the division by $\B$, which we list 
below. \\
i) $M_1 = t_1  \cdots  \widehat{t_j}  \cdots  t_d$; \\
ii) $M_2 =  t_j^{q - 2} t_{c_1}  \cdots  t_{c_u}.t_{b_1} \cdots  
t_{b_v}$, where 
\[
\{t_{c_1} , \ldots , t_{c_u}\} \subset \{t_1, \ldots, t_d\}\setminus \{ t_j\}, 
\;\; \{t_{b_1}, \ldots , t_{b_v}\} \subset \{t_{d + 1}, \ldots, 
t_s\},
\]
with $u + v = d$, $0 \leq u \leq d - 2$ (because there exists $t_\ell$ in 
$\{t_1, \ldots, t_s\} \setminus \{t_j\}$ such that $t_\ell \nmid M_2$) and $2 
\leq v \leq d$; \\
iii) $M_3 = t_j^{q - 2} t_e^2t_{c_1}  \cdots  t_{c_u}.t_{b_1} \cdots  
t_{b_v}$, where 
\[
\{t_{c_1} , \ldots , t_{c_u}\} \subset \{t_1, \ldots, t_d\}\setminus \{ t_j\}, 
\;\; \{t_e\} \cup  \{t_{b_1}, \ldots , t_{b_v}\} \subset \{t_{d + 1}, \ldots, 
t_s\},
\]
$t_e$ distinct from $t_{b_1}, \cdots , t_{b_v}$, with $u + v = d - 2$, $0 \leq 
u \leq d - 2$ and $0 \leq v \leq d - 
2$;\\
iv) $M_4 = t_1  \cdots  \widehat{t_j}  \cdots  t_d t_{e_1}^{q - 2} t_{e_2}$,
where $t_{e_1}$ and $t_{e_2}$ are distinct monomials in the set $\{t_{d + 1}, 
\ldots, t_s\}$.

As observed in the beginning of Section 3 (see inequality \eqref{bound-2nd})
for each possibility of leading monomial $M_i$, $i = 1, \ldots, 4$ 
we want to count the number of monomials which are multiple of $M_i$ but are 
not 
multiple of $\lm(f)$, because then we obtain the following lower bound for 
the weight of $\varphi(f + I_X)$:
\begin{equation} \label{lower-bound}
 \begin{split}  
\omega(\varphi(f + I_X)) 
&\geq  (q - 2)^d (q - 1)^{s - d} \\
&+ |\{M \in \Delta(I_X) \mid  M 
\textrm{ is a multiple of } M_i \\
&\textrm{\hspace{25ex}} \textrm{ and not a multiple of } \lm(f) \}|.
\end{split}
\end{equation}

\begin{lemma}\label{num_mon}
Let $M_i$, $i = 1, \ldots, 4$ as above, and let $N(M_i)$ be the number  of 
monomials which are multiples of $M_i$ and not 
multiples of $\lm(f) = t_1 \cdots  t_d$. Then\\
i) $N(M_1) = (q - 2)^{d -1} (q - 1)^{s - d}$; \\
ii) $N(M_2) = (q - 2)^d (q - 1)^{s - d - v}\left( (q - 1)^{v - 1} - (q - 
2)^{v - 1}\right)$;\\
iii) $N(M_3) = (q - 3)(q - 2)^{d - 2}(q - 1)^{s - d - 1 - v}\left( (q - 
1)^{v + 1} - (q - 2)^{v + 1}\right)$; \\
iv) $N(M_4)= (q - 2)^{d}(q - 1)^{s - d - 2}$ .

\end{lemma}
\begin{proof}
i) From $M_1 = t_1  \cdots  \widehat{t_j}  \cdots  t_d$ we easily get that   
 $N(M_1) = (q - 2)^{d -1} (q - 1)^{s - d}$. \\
ii) We have $M_2 = t_j^{q - 2} t_{c_1}  \cdots  t_{c_u}.t_{b_1} \cdots  
t_{b_v}$, with $u + v = d$, $0 \leq u \leq d - 2$ and $2 \leq v \leq d$. The 
number of monomials which are multiples of $M_2$ is then 
$N_1 = (q - 2)^d(q - 1)^{s - d - 1}$. The set of monomials which are multiples 
of $M_2$ and multiples of $\lm(f)$ is exactly the set of monomials 
which are 
multiples of $t_1  \cdots  t_j^{q - 2} \cdots t_d.t_{b_1} \cdots  
t_{b_v}$, so the number of such monomials is $N_2 = (q - 2)^{d - 1}(q - 2)^v 
(q - 1)^{s - d - v} = (q-2)^{d - 1 + v} (q - 1)^{s - d - v}$. Thus 
\[
N(M_2) = N_1 - N_2 = (q - 2)^d (q - 1)^{s - d - v}\left( (q - 1)^{v - 1} - 
(q - 
2)^{v - 1}\right).
\]
\\
iii) We have $M_3 = t_j^{q - 2} t_e^2t_{c_1}  \cdots  t_{c_u}.t_{b_1} \cdots 
 t_{b_v}$, with $u + v = d - 2$, $0 \leq u \leq d - 2$ and $0 \leq v \leq d - 
2$. The 
number of monomials which are multiples of $M_3$ is $\tilde{N}_1 = (q - 3)(q - 
2)^{d - 
2}(q - 1)^{s - d}$.  The set of monomials which are multiples 
of $M_3$ and multiples of $\lm(f)$ is exactly the set of monomials 
which 
are multiples of $t_1  \cdots  t_j^{q - 2} \cdots t_d. t_e^2.t_{b_1} 
\cdots  t_{b_v}$, so the number of such monomials is 
\begin{equation*}
\begin{split}
\tilde{N}_2 &= (q - 2)^{d 
- 1}(q - 3)(q - 2)^v (q - 1)^{s - d - 1 - v} \\ 
&= (q - 3)(q - 2)^{d - 1 + v} (q 
-1)^{s - d - 1 - v}
\end{split}
\end{equation*}
So we get 
\[
N(M_3) = \tilde{N}_1 - \tilde{N}_2 = (q - 3)(q - 2)^{d - 2}(q - 1)^{s - d - 
1 - v}\left( (q - 
1)^{v + 1} - (q - 2)^{v + 1}\right).
\] \\
iv) We have $M_4 = 
t_1  \cdots  \widehat{t_j}  \cdots  t_d t_{e_1}^{q - 2} t_{e_2}$, and in 
this 
case it is simple to check that $N(M_4) = (q - 2)^{d - 1}(q - 2) (q - 1)^{s - d 
- 
2} 
=   (q - 2)^{d}(q - 1)^{s - d - 2}$.
\end{proof}

We now want to determine which of the values in Lemma \ref{num_mon}  is the 
least one.
We rewrite $N(M_2)$ and $N(M_3)$ as $N(M_2, v)$ and $N(M_3, v)$ to make 
explicit the 
effect of the value of $v$ in these numbers.

\begin{proposition}\label{nm4-minimal}
We have 

\begin{equation*}
\begin{split}
N(M_4) \leq \min ( \{ N(M_1) \} &\cup \{N(M_2, v) \mid  2 \leq v \leq 
d\} \\
      &\cup \{ N(M_3, v) \mid  1 \leq v \leq d - 2\} )
\end{split}
\end{equation*}
while $N(M_3, 0) < N(M_4)$.
\end{proposition}
\begin{proof}
Clearly $N(M_4) < N(M_1)$. From the formula for $N(M_2, v)$ we see that its 
value 
increases as the value of $v$ increases, and the least value for $v$ is $2$, so 
the least value for $N(M_2, v)$ is $N(M_2, 2) = (q - 2)^d (q - 1)^{s - d - 2}$, 
which is equal to $N(M_4)$.

Let $\tilde{N}_2$ be as in the proof of Lemma \ref{num_mon}, clearly 
$\tilde{N}_2$ 
decreases as $v$ increases, which means that $N(M_3, v)$ increases as 
$v$ increases, so the least number for $N(M_3, v)$ is 
$N(M_3, 0) = (q - 3)(q - 2)^{d - 2}(q - 1)^{s - d - 1}$.
We have 
\[
N(M_4) - N(M_3, 0) =  (q - 2)^{d - 2}(q - 1)^{s - d - 2}
\]
so $N(M_3, 0) < N(M_4)$. On the other hand
\[
N(M_3, 1) - N(M_4) = (q - 2)^{d - 2}(q - 1)^{s - d -2}(q^2 - 5q +5)
\]
so (recalling that $q \geq 4$) we get that $N(M_4) < N(M_3, 1)$ and, a 
fortiori, $N(M_4) < N(M_3, v)$ for all $v \in \{1, \ldots, d-2\}$.
\end{proof}
\vspace{2ex}

%

Thus, the lowest bound for $\omega(\varphi(f + I_X))$, according to 
\eqref{lower-bound} is 
\[
\omega(\varphi(f + I_X)) \geq  (q - 2)^d (q - 1)^{s - d} + N(M_3,0),
\]
which could be attained only by a polynomial $f$ such that, for some $j \in 
\{1,\ldots,d\}$, has  $M_3$, with $v = 0$,
as the leading monomial of the remainder in the division of 
$S(t_j^{q - 1} - 1, f)$ by $\B$. 
Yet, a consequence of the next result is that there's no 
homogeneous square-free monic polynomial of degree $d$ with $\lm(f) = 
t_1  \cdots  t_d$ such that
$\omega(\varphi(f + I_X)) = (q - 2)^d (q - 1)^{s - d} + N(M_3,0)$.

\begin{remark}
Note that $M_3$, with $v = 0$, is of the form 
\[
M_3 = t_j^{q - 2} t_{e_1}^2t_{c_1}  \cdots  t_{c_{d - 2}},
\]
where $\{t_{c_1} , \ldots , t_{c_{d - 2}}\} \subset \{t_1, 
\ldots, t_d\}\setminus \{ t_j\}$ and $t_{e_1} \in \{t_{d + 1}, \ldots, t_s\}$. 
Equivalently 
\[
M_3 = t_j^{q - 2} t_{e_1}^2 t_1  \cdots  \hat{t}_j  \cdots  \hat{t}_\ell  
\cdots  t_d
\] 
for some $\ell \in \{1,\ldots, d\} \setminus \{ j \}$.
For the remainder in the division of  
$S(t_j^{q - 1} - 1, f)$ by $\B$
to have  such a monomial as leading monomial, 
according to the proof of Theorem \ref{rem-mon}, we must have: \\
a) $h_1 \neq 0$, and we write $h_1 = \sum_{i = 1}^{w} a_i t_{e_i}$, where 
$t_{e_1} \succ \cdots \succ t_{e_w}$ and $a_1, 
\ldots, a_w \in \fq^*$; \\
b) $t_{e_1} \mid \lm(h_{d - 1})$. \\  
\end{remark}


\begin{proposition}\label{m3}
Let $f$ be 
a homogeneous square-free monic polynomial of degree $d$ such that $\lm(f) = 
t_1  \cdots  
t_d$, and let $j \in \{1, \ldots, d\}$. If the remainder in the division of 
$S(t_j^{q - 1} - 1, f)$ by $\B = \{t_i^{q - 1} - 1 \mid i = 1, \ldots, s\} \cup \{ 
f\}$ has 
\[
M_3 = t_j^{q - 2} t_{e_1}^2 t_1  \cdots  \hat{t}_j  \cdots  \hat{t}_\ell  
\cdots  t_d
\]
as leading monomial, where $\ell \in \{1,\ldots, d\} \setminus \{ j \}$, then there exists $j' \in \{1, \ldots, d\} \setminus \{ j \}$ such that the remainder  in the division of 
$S(t_{j'}^{q - 1} - 1, f)$ by  $\B$ is not zero. Moreover, if $M$ is the 
leading 
monomial of this remainder, then the number of monomials  which are 
multiples of $M$ and are not multiples of $M_3$ or $t_1 \cdots  t_d$ is 
greater than $(q - 2)^{d - 2}(q - 1)^{s - d - 2}$.
\end{proposition}
\begin{proof}
Let $f$ be as in the statement, and assume that for all $j' \in \{1, \ldots, 
d\} \setminus \{ j \}$ the remainder  in the division of 
$S(t_{j'}^{q - 1} - 1, f)$ by  $\B$ is zero. From Corollary \ref{rem-eq-zero} 
we get that 
\[
f = \left(\prod_{\substack{i=1 \\ i\neq j}}^d (t_i + a_{c_i} t_{c_i}) \right) 
f_1
\]
where $t_{c_i} \in \{t_{d+1}, \ldots , t_s\}$ and $a_{c_i} \in \fq^*$ for all 
$i \in \{1, \ldots, d\}\setminus \{ j \}$, and $f_1$ is a homogeneous 
polynomial of degree one, with $\lm(f_1) = t_j$.

Writing $f$ as in Equation \eqref{f}, namely 
\begin{equation*}
f =  t_1  \cdots  t_d + t_j h_{d - 1} + t_1  \cdots  \widehat{t_j}  \cdots 
t_d h_1 + h_d 
\end{equation*}
and comparing with the above expression, we get that $f_1 = t_j + h_1$ and 
\[
\prod_{\substack{i=1 \\ i\neq j}}^d (t_i + a_{c_i} t_{c_i}) = t_1 \cdots  
\hat{t}_j  \cdots  t_d + h_{d - 1} .
\]
From the hypothesis and item (a) of the above Remark we  should have 
$h_1 = \sum_{i = 1}^{w} a_i t_{e_i}$. But then item (b) cannot hold because 
$t_{c_i} \neq t_{e_1}$ for all $i \in \{1, \ldots, d\}\setminus \{ j \}$, since 
$f$ is square-free. 

This proves that there exists $j' \in \{1, \ldots, d\} \setminus \{ j \}$ such 
that the remainder  in the division of 
$S(t_{j'}^{q - 1} - 1, f)$ by  $\B$ is not zero, and
let $M$ be the leading monomial of this remainder. From Theorem \ref{rem-mon} 
we know that $M$ is one of the four types written in a more detailed manner 
before Lemma 
\ref{num_mon}. 
To find the number $N$ of monomials which are multiples of $M$ and are not 
multiples of either $M_3$ or $t_1 \cdots  t_d$, we will determine the number 
$N_1$ of multiples of $M$, then the number $N_2$ of multiples of $M$ and $M_3$, 
then the number $N_3$ of multiples of $M$ and $t_1 \cdots  t_d$, and finally 
the number $N_4$ of common multiples of $M$, $M_3$ and $t_1 \cdots  t_d$. 
From the inclusion-exclusion principle we get that $N = N_1 - N_2 - N_3 + N_4$.
We do this for the four types of monomials $M$. The calculations are not 
difficult but are a bit lengthy, so we present them explicitly for the case 
when $M$ is of type $M_4$, and we will just state the results for the other 
types. 

If $M = t_1  \cdots  \widehat{t_{j'}}  \cdots  t_d t_{e'_1}^{q - 2} 
t_{e'_2}$,
where $t_{e'_1}$ and $t_{e'_2}$ are distinct monomials in the set $\{t_{d + 
1}, 
\ldots, t_s\}$, then $N_1 = (q - 2)^{d - 1}  (q - 2) (q - 1)^{s - d - 1} = (q - 
2)^d (q - 1)^{s - d - 1}$. The number $N_2$ of multiples of $M$ and 
\[
M_3 = t_j^{q - 2} t_{e_1}^2 t_1  \cdots  \hat{t}_j  \cdots  \hat{t}_\ell 
 \cdots  t_d
\]
is equal to the number of multiples of \\
$t_j^{q - 2} t_{e_1}^2 t_{e'_1}^{q - 2} t_{e'_2} t_1  \cdots  \hat{t}_j  
\cdots    t_d$ if $\ell \neq j'$,  $e_1 \neq e'_1$ and $e_1 \neq e'_2$; or 
\\
$t_j^{q - 2}  t_{e'_1}^{q - 2} t_{e'_2} t_1  \cdots   \hat{t}_j  \cdots 
t_d$ if $\ell \neq j'$ and  $e_1 = e'_1$; or \\
$t_j^{q - 2} t_{e'_1}^{q - 2} t_{e'_2}^2 t_1  \cdots  \hat{t}_j  \cdots 
  t_d$
if $\ell \neq j'$ and  $e_1 = e'_2$; or \\
$t_j^{q - 2} t_{e_1}^2 t_{e'_1}^{q - 2} t_{e'_2} t_1  \cdots  \hat{t}_j  
\cdots  \hat{t}_\ell  \cdots   t_d$ if $\ell = j'$,  $e_1 \neq e'_1$ and 
$e_1 \neq e'_2$; or \\
$t_j^{q - 2} t_{e'_1}^{q - 2} t_{e'_2} t_1  \cdots  \hat{t}_j  \cdots  
\hat{t}_\ell  \cdots   t_d$ if $\ell = j'$ and  $e_1 = e'_1$; or \\
$t_j^{q - 2} t_{e'_1}^{q - 2} t_{e'_2}^2 t_1  \cdots  \hat{t}_j  \cdots 
 \hat{t}_\ell  \cdots   t_d$ if $\ell = j'$ and $e_1 = e'_2$.\\
Thus the values of $N_2$ are, respectively, equal to $(q - 3)(q - 2)^d(q - 
1)^{s - d - 3}$, $(q - 2)^d(q - 1)^{s - d - 2}$, $(q - 3)(q - 2)^{d-1}(q - 
1)^{s - d - 2}$, $(q - 3)(q - 2)^{d - 1}(q - 1)^{s - d - 2}$, $(q - 2)^{d-1}(q 
- 1)^{s - d - 1}$ and $(q - 3)(q - 2)^{d - 
2}(q - 1)^{s - d - 1}$. The number $N_3$ of multiples of $M$ and $t_1 \cdots  
t_d$ is equal to the number of multiples of $t_1  \cdots  t_d t_{e'_1}^{q - 
2} t_{e'_2}$, so $N_3 = (q - 2)^{d + 1} (q - 1)^{s - d -2}$. Finally the 
number $N_4$ is equal to the number of multiples of \\
$t_j^{q - 2} t_{e_1}^2 t_{e'_1}^{q - 2} t_{e'_2} t_1  \cdots  \hat{t}_j  
\cdots    t_d$ if  $e_1 \neq e'_1$ 
and $e_1 \neq e'_2$, and in this case $N_4 = (q - 3)(q - 2)^d(q - 1)^{s - d - 
3}$; or \\
$t_j^{q - 2}  t_{e'_1}^{q - 2} t_{e'_2} t_1  \cdots   \hat{t}_j  \cdots 
t_d$ if $e_1 = e'_1$, and then $N_4 = (q - 2)^d(q - 1)^{s - d - 2}$; or \\
$t_j^{q - 2} t_{e'_1}^{q - 2} t_{e'_2}^2 t_1  \cdots  \hat{t}_j  \cdots 
  t_d$ if  $e_1 = e'_2$, and then $N_4 = (q - 3)(q - 2)^{d - 1}(q - 1)^{s - 
d - 2}$.

Now we want to determine $N = N_1 - N_2 - 
N_3 + N_4$, and note that, 
for our purposes, we want the least value of $N$.
Considering all the possible values for $N_2$ and $N_4$ determined above we get 
that this least value is attained when  $\ell = j'$ and  $e_1 = e'_1$, and is 
equal to $(q - 2)^{d - 1} (q - 1)^{s - d - 2}(q - 3)$, so 
when $M$ is of type $M_4$ we get 
\[
N \geq (q - 2)^{d - 1} (q - 1)^{s - d - 2}(q - 3) > (q - 2)^{d - 2} (q - 1)^{s 
- d - 2}.
\]

When $M$ is of type $M_1$, say $M = t_1  \cdots 
\widehat{t_{j'}}  \cdots t_d$, we get $N_1 = (q - 2)^{d - 1}(q - 1)^{s - d + 
1}$, $N_2 = (q - 3)(q - 2)^{d - 2}(q - 1)^{s - d}$ (the greatest of two 
possibilities), $N_3 = (q - 2)^d (q - 1)^{s - d}$ and $N_4 = (q - 3)(q - 2)^{d 
- 1}(q - 1)^{s - d - 1}$, so that $N = (q - 2)^{d - 2}(q - 1)^{s - d - 1}(q^2 - 
4 q + 5) >  (q - 2)^{d - 2}(q - 1)^{s - d - 2}$.  

Now assume that $M$ is of type $M_2$, say $M = t_{j'}^{q - 2} t_{d_1}  \cdots 
 t_{d_k}. t_{b_1}  \cdots  t_{b_v}$, 
where 
\[
\{t_{d_1} , \ldots , t_{d_k}\} \subset \{t_1, \ldots, t_d\}\setminus \{ 
t_{j'}\}, 
\;\; \{t_{b_1}, \ldots , t_{b_v}\} \subset \{t_{d + 1}, \ldots, 
t_s\},
\]
$k + v = d$, $0 \leq k \leq d - 2$ and $2 \leq v \leq d$.
We have $N_1 = (q - 2)^d (q - 1)^{s - d - 1}$. To determine $N_2$ we must 
consider two cases depending on whether or not $t_{e_1}$ belongs to $\{t_{b_1}, 
\ldots , t_{b_v}\}$, 
and each of these two cases subdivide into three other cases.
The expressions that we obtain for $N_2$ depend on $v$. 
Next we have $N_3 = (q - 2)^{d + v - 1}(q - 1)^{s - d - v}$.
For $N_4$ we have to consider the same six cases, as in the determination of 
$N_2$, and the expressions again depend on $v$.  We end up with six expressions 
for $N$, all depending on $v$, and all having their minimum when $v = 2$. Thus, 
taking $v = 2$ we get that the least possible value for $N$ in this case is 
$(q - 2)^{d - 2}(q - 1)^{s - d - 2}(q^2 - 5 q + 7)$, which is greater than 
$(q-2)^{d-2} (q - 1)^{s - d - 2}$.

To finish, we assume that $M$ is of type $M_3$, say 
\[
M = t_{j'}^{q - 2} t_{e'_1}^2t_{d_1}  \cdots  t_{d_k}.t_{b_1} \cdots  
t_{b_v},
\]
 where 
\[
\{t_{d_1} , \ldots , t_{d_k}\} \subset \{t_1, \ldots, t_d\}\setminus \{ 
t_{j'}\}, 
\;\; \{t_{e'_1}\} \cup  \{t_{b_1}, \ldots , t_{b_v}\} \subset \{t_{d + 1}, 
\ldots, 
t_s\},
\]
$t_{e'_1}$ distinct from $t_{b_1}, \cdots , t_{b_v}$, with $k + v = d - 2$, 
$0 \leq k \leq d - 2$ and $0 \leq v \leq d - 2$. We have $N_1 = (q - 3)(q - 
2)^{d - 2}(q - 1)^{s - d}$. For $N_2$, as before, we must consider two cases,
namely   $t_{e_1} \in \{t_{b_1}, \ldots, t_{b_v}\}$, which will branch into 
three 
subcases, and 
$t_{e_1} \notin \{t_{b_1}, \ldots, t_{b_v}\}$, which will branch into six 
subcases. The expressions that we obtain for $N_2$, in each case, depend on 
$v$. 
For $N_3$ we 
have $N_3 = (q - 3)(q - 2)^{d + v - 1} (q - 1)^{s - d - v - 1}$. For $N_4$ 
we must consider the same nine 
cases which appeared in the determination of $N_2$, and again we find 
expressions that depend on $v$. Then we consider the nine expressions that we 
have for $N$, and we observe that each of them 
attain their 
minimum when $v$ has its minimum value: zero, in the cases where 
$t_{e_1} \notin \{t_{b_1}, \ldots, t_{b_v}\}$, and 1 in the cases where
$t_{e_1} \in \{t_{b_1}, \ldots, t_{b_v}\}$. Substituting $v$ by its minimum 
value, in each case, we get that the least value for $N$, when $M$ is 
of type $M_3$, is $N = (q - 3)(q - 2)^{d - 2}(q - 1)^{s - d - 1}$, which is 
greater than $(q - 2)^{d - 2}(q - 1)^{s - d - 2}$.
\end{proof}

Before the above result we had the bound 
\[
\omega(\varphi(f + I_X)) \geq  (q - 2)^d (q - 1)^{s - d} + N(M_3,0),
\]
which, according to Proposition  \ref{nm4-minimal}, could only be attained  
by a polynomial $f$ 
such that, for some $j \in 
\{1,\ldots,d\}$, has  $M_3$, with $v = 0$,
as leading monomial of the remainder in the division of 
$S(t_j^{q - 1} - 1, f)$ by $\B$. Let $g_j$ be this remainder. By the above 
result, there exists $j' \neq j$, $j' \in \{1, \ldots, d\}$ such that the  
remainder in the division of 
$S(t_{j'}^{q - 1} - 1, f)$ by $\B$, which we call $g_{j'}$, is not zero. 
Thus 
\[
I_X + (f) = (t_1^{q -1} - 1, \ldots, t_s^{q - 1} - 1, f, g_j, g_{j'}).
\]
From the definition of footprint, we get that a monomial $M \in \Delta(I_X 
+ (f))$ is not a multiple of any of the monomials in the set 
\[
\{t_1^{q - 1}, 
\ldots, t_s^{q - 1}, \lm(f), \lm(g_j), \lm(g_{j'}) \}.
\]
Thus, reasoning as we did before inequality \ref{bound-2nd}, we get that for 
such $f$ we have
\begin{equation*}
 \begin{split}  
\omega(\varphi(f + I_X)) &\geq 
| \Delta(I_X) | - | \Delta(I_X + (f))| \\
&\geq  (q - 2)^d (q - 1)^{s - d} \\
&+ |\{M \in \Delta(I_X) \mid  M 
\textrm{ is a multiple of } \lm(g_j) \\
&\textrm{\hspace{15ex}} \textrm{ and not a multiple of } \lm(f) \}| \\
&+ |\{M \in \Delta(I_X) \mid  M 
\textrm{ is a multiple of } \lm(g_{j'}) \\
&\textrm{\hspace{15ex}} \textrm{ and not a multiple of } \lm(f) \textrm{ or } 
\lm(g_j) 
\}| 
\end{split}
\end{equation*}
and from the above Proposition we get the strict inequality
\[
\omega(\varphi(f + I_X)) > (q - 2)^d (q - 1)^{s - d} + N(M_3,0) + 
(q - 2)^{d - 2}(q - 1)^{s - d - 2}.
\]
From the proof of Proposition \ref{nm4-minimal} we know that
$N(M_4) - N(M_3,0) = (q - 2)^{d - 2}(q - 1)^{s - d - 2}$, so the above 
inequality 
may be rewritten as
  $\omega(\varphi(f + I_X)) > (q - 2)^d (q - 1)^{s 
- d} + 
N(M_4)$. 

This shows that the lowest bound for the next-to-minimal weight is 
$(q - 2)^d (q - 1)^{s - d} + N(M_4)$ (which can be only be attained by a 
homogeneous square-free monic polynomial of degree $d$ with $\lm(f) = 
t_1  \cdots  t_d$ such that, for some $j \in 
\{1,\ldots,d\}$, has either $M_4$ or $M_2$, with $v = 2$,
as the leading monomial of the remainder in the division of 
$S(t_j^{q - 1} - 1, f)$ by $\B$). The next result shows that this bound is 
attained.

\begin{theorem} \label{main}
Assume that $2d + 2 \leq s$. Then the next-to-minimal weight of the code $C(d)$ 
is $(q - 2)^d (q - 1)^{s - d} + (q - 2)^d(q - 1)^{s - d - 2}$.
\end{theorem}
\begin{proof}
We know that a lower bound for 
the next-to-minimal weight of a codeword is  $\delta := (q - 2)^d (q - 1)^{s - 
d} + (q - 2)^d(q - 1)^{s - d - 2}$. 
Let 
\[
f := \left( \prod_{i = 1}^{d - 1}(t_i - t_{d - 1 + i }) \right) (t_{2d -1} - 
t_{2d} + 
t_{2d + 1} - t_{2d + 2}), 
\]
we claim that $\varphi(f + I_X)$ has weight equal to 
$\delta$. 

We start by noting that for $i \in \{1, \ldots, d - 1\}$ the number of pairs 
$(a_i, a_{d -1 + i}) \in (\fq^*)^2$ such that $a_i - a_{d - 1 + i} \neq 0$ is 
equal 
to 
$(q - 1)^2 - (q - 1) = (q - 2)(q - 1)$. Now we will count the number of 
$4$-tuples $(a_d, a_{2d}, a_{2d + 1}, a_{2d + 2}) \in (\fq^*)^4$
such that $a_{2d - 1} - a_{2d} + a_{2d + 1} - a_{2d + 2} \neq 0$. 
Observe that for every $a \in \fq^*$ the number of pairs $(a_{2d - 1}, a_{2d}) 
\in 
(\fq^*)^2$ such that $a_{2d - 1} - a_{2d} = a$ is equal to $q - 2$, and the 
number of 
pairs $(a_{2d+1}, a_{2d+2}) \in (\fq^*)^2$ such that $a_{2d + 1} - a_{2d + 2} 
\neq -a$ is $(q - 1)^2 - (q - 2)$. Thus, for each $a \in \fq^*$ we have $(q - 
2)((q - 1)^2 - (q - 2))$ $4$-tuples 
such that $a_{2d - 1} - a_{2d} + a_{2d + 1} - a_{2d + 2} \neq 0$, with $a_{2d - 
1} - a_{2d} = 
a$, which gives us $(q - 2)(q -1)((q - 1)^2 - (q - 2))$ $4$-tuples such that
$a_{2d - 1} - a_{2d} + a_{2d + 1} - a_{2d + 2} \neq 0$ and $a_{2d - 1} - a_{2d} 
\neq 0$.
On the other hand, we have $q-1$ pairs 
$(a_{2d - 1}, a_{2d}) \in (\fq^*)^2$ such that $a_{2d - 1} - a_{2d} = 0$ and 
$(q-2)(q-1)$ 
pairs $(a_{2d+1}, a_{2d+2}) \in (\fq^*)^2$ such that
$a_{2d+1} - a_{2d+2} \neq 0$, 
so the total number of $4$-tuples $(a_{2d - 1}, a_{2d}, a_{2d + 1}, a_{2d + 2}) 
\in 
(\fq^*)^4$
such that $a_{2d - 1} - a_{2d} + a_{2d + 1} - a_{2d + 2} \neq 0$ is
\[
 (q - 2)(q -1)((q - 1)^2 - (q - 2)) + (q - 2)(q - 1)^2 = (q - 2)(q - 1)^3 + (q 
 - 2)(q - 1).
\]

Hence the number of $(2d + 2)$-tuples $(a_1, \ldots, a_{2 d + 2}) \in 
(\fq^*)^{2 d + 2}$ such that 
\begin{equation}\label{notzero}
\left( \prod_{i = 1}^{d - 1}(a_i - a_{d - 1+ i}) \right) (a_{2d - 1} - a_{2d} + 
a_{2d + 1} - a_{2d + 2}) \neq 0 
\end{equation}
is $((q-2)(q - 1))^{d - 1}((q - 2)(q - 1)^3 + (q - 2)(q - 1)) = (q - 2)^d(q - 
1)^{d + 2} + (q - 2)^d (q - 1)^d$.

We conclude that the number of $s$-tuples $(a_1, \ldots, a_s) \in (\fq^*)^s$ 
such that inequality \eqref{notzero} holds is 
\begin{equation*} 
\begin{split}
(q - 1)^{s - 2d -2}&((q - 2)^d(q - 1)^{d + 2} + (q - 2)^d (q - 1)^d) = \\
 &(q - 2)^d (q - 1)^{s - d} + (q - 2)^d(q - 1)^{s - d - 2}.
\end{split}
\end{equation*} 
\end{proof}

\begin{corollary}
Assume that $2 d - 2 \geq s$. Then the next-to-minimal weight of 
the code $C(d)$ 
is $(q - 2)^{s - d} (q - 1)^{d} + (q - 2)^{s - d}(q - 1)^{d - 2}$.
\end{corollary}
\begin{proof}
This is a consequence of the isomorphism between $C(d)$ and $C(s - d)$ which 
holds for $d$ such that $2 d \leq s$ (see the proof of Theorem 4.5 in 
\cite{evalcodes}). From $2d  - 2 \geq s$ we get $2d + s \geq 2 s + 2$ so that 
$2(s - d) + 2 \leq s$. Thus we may apply Theorem \ref{main} and get that the 
next-to-minimal weight of $C(s -d)$ (and hence, of $C(d)$) is 
$(q - 2)^{s - d} (q - 1)^{d} + (q - 2)^{s - d}(q - 1)^{d - 2}$.
\end{proof}

\begin{remark}
As a consequence of the isomorphism cited in the above proof and the form of 
the polynomial $f$ which appears in the proof of Theorem \ref{main}, we get 
that, when $2 d - 2 \geq s$, the polynomial 
\begin{equation*} 
\begin{split}
g = &\left( \prod_{i = 1}^{s - d - 1} (t_i - t_{s - d -1 + i})\right)(t_{2s - 
2d} 
t_{2s - 2d + 1}t_{2s - 2d + 2} - t_{2s - 2d - 1}t_{2s - 2d + 1} t_{2s - 2d + 2} 
\\ 
&+ t_{2s - 2d - 1}t_{2s - 2d} t_{2s - 2d + 2} - t_{2s - 2d - 1}t_{2s - 2d} 
t_{2s 
- 2d + 1}) t_{2s - 2d + 3}  \cdots  t_s
\end{split}
\end{equation*}
has degree $d$ and $\omega(\varphi(g + I_X)) = (q - 2)^{s - d} (q - 1)^{d} + (q 
- 2)^{s - d}(q - 1)^{d - 2}$ (in the formula of $g$ the product 
$t_{2s - 2d + 3}  \cdots  t_s$ only appears if $2 d - 2 > s$).
\end{remark}

\begin{remark}
A couple months after we submitted this paper we found that, for $q \geq 5$,  
the value of the
next-to-minimal weight  is  
$(q - 2)^{d - 1} (q - 1)^d ((q - 2)^2 + (q - 1))$, 
when $s = 2d + 1$, 
and is    
$(q - 2)^{d - 2} (q - 1)^{d-1} ((q - 2)^2 + (q - 1))$ 
 when $s = 2d - 1$. We are still working on the
case where $s = 2d$.
\end{remark}

\vspace{2ex}
\noindent
{\small
\textbf{Acknowledgments.} We started this paper during the International 
Conference 
on Algebraic Geometry, Coding Theory and Combinatorics, held at IIT-Hyderabad, 
India in December of 2023. We want to thank the organizers of this meeting for 
the nice scientific environment they created, and for partially supporting our 
participation in it. We also thank Prof.\ \'{E}rika Lopes for enlightening 
conversations on the combinatorics of this paper. 

\noindent
C. Carvalho  was partially supported by Fapemig APQ-01430-24 and CNPq PQ 
308708/2023-7 \\
N. Patanker was partially supported by an IoE-IISc Postdoctoral fellowship.}


\begin{thebibliography}{M}
\bibitem{bruno} B.\ Buchberger, Ein Algorithmus zum Auffinden der Basiselemente 
des Restklassenringes nach einem nulldimensionalen Polynomideal. Mathematical 
Institute, University of Innsbruck, Austria. PhD Thesis. 1965. An English 
translation appeared in J.\ Symbolic Comput.\ 41 (2006) 475-511.

\bibitem{car-2013}
Carvalho, C. On the second Hamming weight of some Reed-Muller type codes. 
Finite Fields Appl. 24 (2013), 88--94.

\bibitem{gb-in-coding}
Carvalho, C. Gr\"obner bases methods in coding theory. Algebra for secure and 
reliable communication modeling, 73--86, Contemp. Math., 642, Amer. Math. Soc., 
Providence, RI, 2015.

\bibitem{car-neu-2017}
Carvalho, C; Neumann, V.G.L. On the next-to-minimal weight of affine cartesian 
codes. Finite Fields Appl. 44 (2017), 113--134. 

\bibitem{car-2024}
Carvalho, Cícero; López, Hiram H.; Matthews, Gretchen L. Decreasing norm-trace 
codes. Des. Codes Cryptogr. 92 (2024), no. 5, 1143--1161.

\bibitem{cox}
Cox, D.; Little, J.; O'Shea, D. Ideals, Varieties, and Algorithms, 3rd ed., 
Springer, New York, 2007.


\bibitem{fl} Fitzgerald, J.;  Lax, R.F.
Decoding affine variety codes using Gr\"obner bases,
Des.\ Codes and Cryptogr. \textbf{13}(2) (1998) 147--158.

\bibitem{geil}
Geil, O.; Høholdt, T. Footprints or generalized Bezout's theorem. IEEE Trans. 
Inform. Theory 46 (2000), no. 2, 635--641.

\bibitem{geil2}
Geil, Olav. On the second weight of generalized Reed-Muller codes. Des. Codes 
Cryptogr. 48 (2008), no. 3, 323--330.

\bibitem{hansen}Hansen, J.P. Toric surfaces and error-correcting codes, in: 
Coding Theory, 
 Cryptography and Related Areas, Guanajuato, 1998, 
Springer, Berlin, 2000, pp. 132–142.

\bibitem{evalcodes} Jaramillo, Delio; Vaz Pinto, Maria; Villarreal, Rafael H. 
Evaluation codes and their basic parameters. Des. Codes Cryptogr. 89 (2021), 
no. 2, 269--300.

\bibitem{jaramillo2023}
Jaramillo-Velez, Delio; López, Hiram H.; Pitones, Yuriko. Relative generalized 
Hamming weights of evaluation codes. São Paulo J. Math. Sci. 17 (2023), no. 1, 
188--207. 


\bibitem{rolland}
Rolland, R. The second weight of generalized Reed-Muller codes in most cases. 
Cryptogr. Commun. 2 (2010), no. 1, 19--40.




\end{thebibliography}
\end{document}